\documentclass{llncs}

\usepackage{amsmath}
\usepackage{mathtools}

\usepackage{ifxetex,ifluatex}
\ifxetex
\usepackage{xltxtra}
\else
\ifluatex
\usepackage{luatextra}
\else
\usepackage{amssymb}
\usepackage[utf8]{inputenc}
\usepackage[T1]{fontenc}
\csname else\expandafter\endcsname
\romannumeral -`0% no space till \relax
\fi
\fi
\iftrue\relax%
\defaultfontfeatures{Ligatures=TeX}
\usepackage[math-style=TeX, vargreek-shape=TeX]{unicode-math}
\AtBeginDocument{%
  \let\setminus\smallsetminus% no U+29F5 in TG Pagella Math
}
\fi

\usepackage{graphicx,color,authblk}
\usepackage[authoryear]{natbib}
\usepackage{enumitem}
\usepackage{xparse}

\usepackage[bookmarksnumbered,bookmarksopen,
unicode]{hyperref}

\DeclareMathOperator{\psdrk}{rank_{PSD}} % PSD rank
\newcommand{\R}{\mathbb{R}}
\newcommand{\N}{\mathbb{N}}
\newcommand{\Z}{\mathbb{Z}}
\newcommand{\cube}[1]{[0,1]^{#1}}

\DeclareMathOperator{\xcs}{xc_{\text{SDP}}}

\newcommand{\face}[1]{\left\{#1\right\}}
\newcommand{\lspan}[1]{\mathrm{span}\left\{#1\right\}}
\newcommand{\vol}[1]{\operatorname{vol}\left(#1\right)}
\newcommand{\conv}[1]{\operatorname{conv}\left(#1\right)}
\newcommand {\card}[1]{\left|#1\right|}

\newcommand{\opnorm}[1]{\left\|#1\right\|}
\newcommand{\norm}[1]{\left\|#1\right\|_2}
\newcommand{\fnorm}[1]{\left\|#1\right\|_F}
\newcommand{\mnorm}[1]{\left\|#1\right\|_\infty}
\newcommand{\lmax}[1]{\left\|#1\right\|_\infty}

\newcommand{\pr}[2]{\langle{#1, #2}\rangle}

\newcommand{\binSet}{\{0,1\}}

\newcommand{\psd}{\mathbb S_+}
\newcommand{\tr}[1]{\mathrm{Tr}\left[#1\right]}
\newcommand{\trans}{\mathsf T}

\newcommand{\beq}{\begin{equation}}
\newcommand{\eeq}{\end{equation}}
\newcommand{\beqn}{\begin{equation*}}
\newcommand{\eeqn}{\end{equation*}}
\newcommand{\beqr}{\begin{eqnarray}}
\newcommand{\eeqr}{\end{eqnarray}}
\newcommand{\beqrn}{\begin{eqnarray*}}
\newcommand{\eeqrn}{\end{eqnarray*}}

\newcommand{\eps}{\varepsilon}
\newcommand{\sphere}[1]{S^{#1 - 1}}

\newcommand{\st}{\big|\,}

\DeclareMathOperator{\Tr}{Tr}

\renewenvironment{proof}[1][]{
  \begin{trivlist}
   \item[\hspace{\labelsep}{\em\noindent Proof#1:\/}]}
   {{\hfill$\Box$}
  \end{trivlist}}

\newenvironment{claimproof}[1][]{
  \begin{trivlist}
   \item[\hspace{\labelsep}{\sc\noindent Proof of claim#1:\/}]}
   {{\hfill$\blacklozenge$}
  \end{trivlist}
}
  
\usepackage{xcolor}

\newtheorem{defn}[theorem]{Definition}

\hypersetup{
  pdftitle={On the existence of 0/1 polytopes with high \\ semidefinite extension complexity},
  pdfauthor={Jop Bri\"et, Daniel Dadush,
    Sebastian Pokutta},
  pdfkeywords={},% should be same as argument of \keywords
  pdfsubject={MSC2000 Primary ;
    Secondary}} % should be "MSC2000" + argument of \subjclass{}

\title{On the existence of 0/1 polytopes with high \\ semidefinite extension complexity%
   % space to separate footnote marks
}

\author{Jop Bri\"et \inst{1}\footnote{J.B. was supported by a Rubicon grant from the Netherlands Organisation for Scientific Research (NWO).} \and Daniel Dadush\inst{1} \and Sebastian
  Pokutta\inst{2}\footnote{Research reported in this paper was partially supported by NSF grant CMMI-1300144.}}
\institute{
New York  University,
 Courant Institute of Mathematical Sciences,
New York, NY, USA.
 \email{\{jop.briet,dadush\}@cims.nyu.edu}
\and
Georgia Institute of Technology,
  H. Milton Stewart School of Industrial and
Systems Engineering,
  Atlanta, GA,
  USA.
  \email{sebastian.pokutta@isye.gatech.edu}}

\begin{document}

\maketitle

\begin{abstract}
In \cite{Rothvoss11} it was shown that there exists a 0/1 polytope (a polytope whose vertices are in~$\{0,1\}^{n}$) such that
any higher-dimensional polytope projecting to it must have $2^{\Omega(n)}$ facets, i.e., its linear extension complexity is exponential. The
question whether there exists a 0/1 polytope with high PSD extension
complexity was left open. 
We answer this question in the
affirmative by showing that there is a 0/1 polytope such that any
spectrahedron projecting to it must be the intersection of a semidefinite cone of dimension~$2^{\Omega(n)}$ and an affine space.
Our proof relies on a new technique to rescale semidefinite factorizations.
\end{abstract}

\section{Introduction}
\label{sec:introduction}
The subject of lower bounds on the size of extended formulations has recently
regained a lot of attention. This is due to several reasons. First of all,
essentially all NP-Hard problems in combinatorial optimization can be expressed
as linear optimization over an appropriate convex hull of integer points.
Indeed, many past (erroneous) approaches for proving that P=NP have proceeded
by attempting to give polynomial sized linear extended formulations for hard
convex hulls (convex hull of TSP tours, indicators of cuts in a graph, etc.).
Recent breakthroughs~\cite{extform4, bfps2012} have unconditionally ruled out
such approaches for the TSP and Correlation polytope, complementing the classic
result of \cite{Yannakakis91} which gave lower bounds for symmetric
extended formulations. Furthermore, even for polytopes over which optimization
is in P, it is very natural to ask what the ``optimal'' representation of the
polytope is. From this perspective, the smallest extended formulation
represents the ``description complexity'' of the polytope in terms of a linear
or semidefinite program.

A \emph{(linear) extension} of a polytope \(P \subseteq \R^n\) is another
polytope \(Q \subseteq \R^d\), so that there exists a linear projection \(\pi\)
with \(\pi(Q) = P\). The \emph{extension complexity} of a polytope is the
minimum number of facets in any of its extensions. The linear extension
complexity of \(P\) can be thought of as the inherent complexity of expressing
\(P\) with linear inequalities. Note that in many cases it is possible to save
an exponential number of inequalities by writing the polytope in higher-dimensional space. Well-known examples include the regular polygon, see
\cite{Ben-TalNemirovski01} and \cite{FioriniRothvossTiwary11} or the
permutahedron, see \cite{Goemans09}. A \emph{(linear) extended formulation} is
simply a normalized way of expressing an extension as an intersection of the
nonnegative cone with an affine space; in fact we will use these notions in an
interchangeable fashion. In the seminal work of \cite{Yannakakis88} a
fundamental link between the extension complexity of a polytope and the
nonnegative rank of an associated matrix, the so called \emph{slack matrix},
was established and it is precisely this link that provided all known strong
lower bounds. It states that the nonnegative rank of any slack matrix is equal
to the extension complexity of the polytope.

As shown in \cite{extform4} and \cite{GouveiaParriloThomas2011} the above
readily generalizes to semidefinite extended formulations. Let \(P \subseteq
\R^n\) be a polytope. Then a \emph{semidefinite extension} of \(P\) is a
spectrahedron \(Q \subseteq \R^d\) so that there exists a linear map \(\pi\)
with \(\pi(Q) = P\). While the projection of a
polyhedron is polyhedral, it is open which convex sets can be obtained as
projections of spectrahedra. We can again normalize the representation by considering
\(Q\) as the intersection of an affine space with the cone of positive semidefinite (PSD) matrices. The
\emph{semidefinite extension complexity} is then defined as the
smallest~\(r\) for which there exists an affine space such that its intersection
with the cone of \(r \times r\) PSD matrices projects to \(P\). We
thus ask for the smallest representation of \(P\) as a projection of a
spectrahedron. In both the linear and the semidefinite case, one can
think of the extension complexity as the  minimum size of the cone needed to
represent \(P\). Yannakakis's theorem can be generalized to this
case, as was done in \cite{extform4} and \cite{GouveiaParriloThomas2011}, and it
asserts that the semidefinite extension complexity of a polytope is equal to
the semidefinite rank (see Definition~\ref{def:sdpfact}) of any of its slack matrices. 

An important fact in the study of extended formulations is that the encoding
length of the coefficients is disregarded, i.e., we only measure the dimension
of the required cone. Furthermore, a lower bound on the extension complexity of a
polytope does not imply that building a separation oracle for the polytope is
computationally hard. Indeed, as recently shown in \cite{Rothvoss13},
the perfect matching polytope has exponential extension complexity,
while the associated separation problem (which
allows us to compute min-cost perfect matchings) is in P. Thus standard
complexity theoretic assumptions and limitations do not apply. In fact one of
the main features of extended formulations is that they \emph{unconditionally}
provide lower bounds for the size of linear and semidefinite programs
\emph{independent of P vs.~NP}. 

The first natural class of polytopes with high linear extension complexity comes
from the work of~\cite{Rothvoss11}. Rothvo\ss\  showed that ``random'' 0/1
polytopes have exponential linear extension complexity via an elegant
counting argument. Given that SDP relaxations are often far more powerful than
LP relaxations, an important open question is whether random 0/1 polytopes also
have high PSD extension complexity.

\subsection{Related work}
\label{sec:related-work} The basis for the study of linear and semidefinite
extended formulations is the work of Yannakakis (see \cite{Yannakakis88} and
\cite{Yannakakis91}). The existence of a 0/1 polytope with exponential
extension complexity was shown in \cite{Rothvoss11} which in turn was inspired
by \cite{Shannon49}. The first explicit example, answering a long standing open
problem of Yannakakis, was provided in \cite{extform4} which, together with
\cite{GouveiaParriloThomas2011}, also lay the foundation for the study of
extended formulations over general closed convex cones. In \cite{extform4} it
was also shown that there exist matrices with large nonnegative rank but small
semidefinite rank, indicating that semidefinite extended formulations can be
exponentially stronger than linear ones, however falling short of giving an
explicit proof. 
They thereby separated
the expressive power of linear programs from those of semidefinite programs and
raised the question:
\begin{center}
{\em Does {\em every} 0/1 polytope have an efficient semidefinite lift?}
\end{center}

Other related work includes \cite{bfps2012}, where the authors study
approximate extended formulations and provide examples of spectrahedra
that cannot be approximated well by linear programs with a polynomial
number of inequalities as well as improvements thereof by \cite{braverman2012information}.
\cite{FaenzaFioriniGrappeTiwary11} proved equivalence of extended formulations to communication complexity. Recently there has been also significant progress in terms of
lower bounding the linear extension complexity of polytopes by means of
information theory, see \cite{braverman2012information} and
\cite{BP2013commInfo}. Similar techniques are not known for the semidefinite case.

\subsection{Contribution}
\label{sec:contribution}

We answer the above question in the negative, i.e., we show the existence of a
0/1 polytope with exponential semidefinite extension complexity. In particular,
we show that the counting argument of \cite{Rothvoss11} extends to the PSD
setting. 

The main challenge when moving to the PSD setting, is that the
largest value occurring in the slack matrix does not easily translate to a
bound on the largest values occurring in the factorizations. Obtaining such a
bound is crucial for the counting argument to carry over.

Our main technical contribution is a new rescaling technique for semidefinite
factorizations of slack matrices. In particular, we show that any rank-$r$ semidefinite
factorization of a slack matrix with maximum entry size~$\Delta$ can be ``rescaled'' to a
semidefinite factorization where each factor has operator norm at most~$\sqrt{r\Delta}$
(see Theorem~\ref{thm:psdRescale}). Here our proof proceeds by a variational argument and
relies on John's theorem on ellipsoidal approximation of convex bodies
\cite{john1948extremum}. We note that in the linear case proving such a result is far
simpler, here the only required observation is that after independent nonnegative
scalings of the coordinates a nonnegative vector remains nonnegative. However, one
cannot in general independently scale the entries of a PSD matrix while maintaining the
PSD property.  

Using our rescaling lemma, the existence proof of the 0/1 polytopes with high semidefinite
extension complexity follows in a similar fashion to the linear case as presented in
\cite{Rothvoss11}.  In addition to our main result, we show the existence of a polygon
with \(d\) integral vertices and semidefinite extension complexity
\(\Omega((\frac{d}{\log d})^{\frac{1}{4}})\). The argument follows similarly to
\cite{FioriniRothvossTiwary11} adapting \cite{Rothvoss11}. 

\subsection{Outline}
\label{sec:outline}

In Section~\ref{sec:preliminaries} we provide basic results and notions. We then present
the rescaling technique in Section~\ref{sec:resc-psd-extend} which is at the core of our
existence proof. In Section~\ref{sec:exist-01-polyt} we establish the existence of 0/1
polytopes with subexponential semidefinite extension complexity and we conclude with some
final remarks in Section~\ref{sec:conclusion}. 

\section{Preliminaries}
\label{sec:preliminaries}

Let \([n] \coloneqq \face{1, \dots, n}\). In the following we will consider semidefinite extended formulations. We refer the interested reader to
\cite{extform4} and \cite{bfps2012} for a broader overview and proofs.  

Let $B_2^n \subseteq \R^n$ denote the $n$-dimensional Euclidean ball, and let $\sphere{n} = \partial B_2^n$ denote the Euclidean sphere in $\R^n$. We
denote by $\psd^n$  the set of $n \times n$ PSD matrices which form a (non-polyhedral) convex cone. Note that~$M \in \psd^n$ if and only if $M$ is
symmetric ($M^\trans = M$) and
\[
x^\trans M x \geq 0 \quad \forall x \in \R^n \text{.}
\]
Equivalently, $M \in \psd^n$ iff $M$ is symmetric and has nonnegative eigenvalues. For a linear subspace $W \subseteq \R^n$, let $\dim(W)$ denote its
dimension, $W^\perp$ its orthogonal complement, and $P_W:\R^n \rightarrow \R^n$ the orthogonal projection onto $W$. Note that as a matrix $P_W \in
\psd^n$ and $P_W^2 = P_W$. For a matrix $A \in \R^{n \times m}$, let ${\rm Im}(A)$ denote its image or column span, and let ${\rm Ker}(A)$ denote its
kernel. For a symmetric matrix $A \in \R^{n \times n}$, we have that ${\rm Im}(A) = {\rm Ker}(A)^\perp$. If $A \in \psd^n$, we have that $x \in {\rm
Ker}(A) \Leftrightarrow x^\trans A x = 0$. We define the pseudo-inverse $A^+$ of a symmetric matrix $A$ to be the unique matrix satisfying $A^+A =
AA^+ = P_W$, where $W = {\rm Im}(A)$. If $A$ has spectral decomposition $A = \sum_{i=1}^k \lambda_i v_iv_i^\trans$, $v_1,\dots,v_k$
orthonormal, then $A^+ = \sum_{i=1}^k \lambda_i^{-1} v_iv_i^\trans$.

For matrices $A,B \in \psd^n$, we have that ${\rm Im}(A + B) = {\rm Im}(A) + {\rm Im}(B)$ and that ${\rm Ker}(A+B) = {\rm Ker}(A) \cap {\rm Ker}(B)$.
We denote the trace of $A \in \psd^n$ by ${\rm Tr}[A] = \sum_{i=1}^n A_{ii}$.  For a pair of equally-sized matrices~$A,B$ we
let~$\langle A,B\rangle = {\mathrm{Tr}}[A^{\trans}B]$ denote their trace inner product and let~$\fnorm{A} = \sqrt{\langle A,A\rangle}$ denote the
Frobenius norm of~$A$. We denote the operator norm of a matrix \(M\in\R^{m\times n}\) by 
\[
\opnorm{M} = \sup_{\norm{x}=1} \norm{Mx} \text{.}
\]
If $M$ is square and symmetric ($M^{\trans} = M$), then $\opnorm{M} = \sup_{\norm{x}=1} |x^\trans Mx|$, in which case $\opnorm{M}$ denotes the largest
eigenvalue of $M$ in absolute value. Lastly, if $M \in \psd^n$ then $\opnorm{M} = \sup_{\norm{x}=1} x^\trans Mx$ by nonnegativity of the inner
expression.

For every positive integer~$\ell$ and any $\ell$-tuple of matrices ${\mathbf M} = (M_1,\dots,M_\ell)$ we define
\beqn
\lmax{\mathbf M} = \max\{\opnorm{M_i}\st i\in[\ell]\}.
\eeqn

\begin{defn}[Semidefinite extended formulation]
Let $K\subseteq \R^n$ be a convex set. A {\em semidefinite extended
  formulation} (semidefinite EF)  of~$K$ is a system consisting of a positive integer~$r$, an index set~$I$ and a set of triples $(a_i,U_i,b_i)_{i\in I}\subseteq \R^n\times \psd^r\times \R$ such that
\beqn
K = \{x\in\R^n\st \exists Y\in \psd^r:\: a_i^{\trans}x + \langle U_i,Y\rangle = b_i\: \forall i\in I\}.
\eeqn
The {\em size} of a semidefinite EF is the size~$r$ of the positive semidefinite matrices~$U_i$.
The {\em semidefinite extension complexity} of~$K$, denoted~$\xcs(K)$,
is the minimum size of a semidefinite EF of~$K$.
\end{defn}

In order to characterize the semidefinite extension complexity of a
polytope $P\subseteq[0,1]^n$ we will need the concept of a slack
matrix.

\begin{defn}[Slack matrix]\label{def:slack}
Let $P\subseteq \cube{n}$ be a polytope, $I,J$ be finite sets, $\mathcal A = (a_i,b_i)_{i\in I}\subseteq \R^n\times \R$ be a set of pairs and let~$\mathcal X = (x_j)_{j\in J}\subseteq\R^n$ be a set of points, such that
\beqn
P = \{x\in\R^n\st a_i^{\trans}x\leq b_i\: \forall i\in I\} = \conv{\mathcal X}.
\eeqn
Then, the {\em slack matrix} of~$P$ associated with~$(\mathcal
A,\mathcal X)$ is given by $S_{ij} = b_i - a_i^{\trans}x_j$.
\end{defn}

Finally, the definition of a semidefinite factorization is as follows.

\begin{defn}[Semidefinite factorization]\label{def:sdpfact}
Let $I, J$ be finite sets,  $S\in\R_+^{I\times J}$ be a nonnegative matrix and~$r$ be a positive integer.
Then, a {\em rank-$r$ semidefinite factorization} of~$S$ is a set of pairs~$(U_i,V^j)_{(i,j)\in I\times J}\subseteq \psd^r\times\psd^r$ such that $$S_{ij} = \langle U_i,V^j\rangle$$ for every~$(i,j)\in I\times J$.
The {\em semidefinite rank} of~$S$, denoted~$\psdrk(S)$, is the minimum~$r$ such that there exists a rank~$r$ semidefinite factorization of~$S$.
\end{defn}

Using the above notions the semidefinite extension complexity of a
polytope can be characterized by the semidefinite rank of any of its slack
matrices, which is a generalization of Yannakakis's factorization
theorem (\cite{Yannakakis88} and \cite{Yannakakis91}) 
established in \cite{extform4} and
\cite{GouveiaParriloThomas2011}.

\begin{theorem}[Yannakakis's Factorization Theorem for SDPs]
\label{thm:yannaSDP}
Let $P\subseteq\cube{n}$ be a polytope and $\mathcal A =
(a_i,b_i)_{i\in I}$ and $\mathcal X = (x_j)_{j\in J}$ be as in Definition~\ref{def:slack}. 
Let $S$ be the slack matrix of~$P$ associated with~$(\mathcal A,\mathcal X)$.
Then, $S$ has a rank-$r$ semidefinite factorization if and only if~$P$
has a semidefinite EF of size~$r$. That is, $\psdrk(S) = \xcs(P)$.

Moreover, if $(U_i,V^j)_{(i,j)\in I\times J}\subseteq \psd^r\times\psd^r$ is a factorization of~$S$, then
\beqn
P = \{x\in\R^n\st \exists Y\in\psd^r:\: a_i^{\trans}x + \langle U_i,Y\rangle = b_i\: \forall i\in I\}
\eeqn
and the pairs~$(x_j,V^j)_{j\in J}$ satisfy~$a_i^{\trans}x_j + \langle U_i,V^j\rangle = b_i$ for every~$i\in I$.

In particular, the extension complexity is independent of the choice
of the slack matrix and the semidefinite rank of all slack matrices of \(P\) is
identical. 
\end{theorem}

The following well-known theorem due to \cite{john1948extremum} lies at the core of our
rescaling argument. We state a version that is suitable for the later application.  
Recall that \(B_2^n\) denotes the \(n\)-dimensional Euclidean unit ball.
A {\em probability vector} is a vector $p\in\R^n_+$ such that $p(1) + p(2) + \cdots + p(n) = 1$.
For a convex set $K \subseteq \R^n$, we let ${\rm aff}(K)$ denote the affine hull of $K$,
the smallest affine space containing $K$. We let $\dim(K)$ denote the linear dimension of
the affine hull of $K$. Last, we let ${\rm relbd}(K)$ denote the relative boundary of $K$,
i.e., the topological boundary of $K$ with respect to its affine hull ${\rm aff}(K)$.

\begin{theorem}[\cite{john1948extremum}]\label{thm:John}
Let \(K \subseteq \R^n\) be a centrally symmetric convex set with $\dim(K)=k$. Let \(T \in
\R^{n \times k}\) be such that \(E = T \cdot B_2^k = \{Tx\st \|x\| \leq 1\}\) is the
smallest volume ellipsoid containing \(K\). Then, there exist a finite set of
points~$\mathcal Z\subseteq {\rm relbd}(K) \cap {\rm relbd}(E)$ and a probability vector
$p\in\R^\mathcal Z_+$ such that 
\[\sum_{z\in\mathcal Z} p(z)\, zz^\trans  = \frac{1}{k} TT^\trans.\]
\end{theorem}

For a real-valued function $f:\R \rightarrow \R$, we denote its right-sided derivative at~$a\in\R$ by 
\[
\frac{d_+}{dx} f |_{x=a} = \lim_{\eps \rightarrow 0^+} \frac{f(a+\eps)-f(a)}{\eps} \text{.}
\]
We will need the following lemma; for  a general theory on
perturbations on linear operators we refer the reader to
\cite{kato1995perturbation}. 

\begin{lemma}\label{lem:derlemma}
\label{prop:standard}  Let $r$ be a positive integer, $X\in \psd^r$ be a non-zero positive semidefinite matrix. Let $\lambda_1 = \opnorm{X}$ and  
$W$ denote the $\lambda_1$-eigenspace of $X$. Then for $Z\in\R^{r\times r}$ symmetric,  
\[
\left.\frac{d_+}{d\eps}\opnorm{X+\eps Z}\right|_{\eps = 0} = \max_{\substack{w \in W \\ \norm{w}=1}} w^\trans  Z w
\] 
\end{lemma}
\begin{proof}
Observe that
\begin{align*}
\opnorm{X + \varepsilon Z} &\geq \max_{\norm{w}=1,  w \in W} w^\trans(X +
\varepsilon Z)w = \max_{\norm{w}=1,  w \in W} \underbrace{w^\trans
  Xw}_{= \lambda_1} + \varepsilon w^\trans Zw \\
&= \lambda_1 + \varepsilon \cdot \max_{\norm{w}=1,  w \in W} w^TZw. 
\end{align*}
It therefore suffices to show that \(\opnorm{X + \varepsilon Z}\)
cannot exceed the lower bound by more than $o(\varepsilon)$.

Let \(u\) be an arbitrary vector with \(\norm{u} = 1\) and write \(u = u_1 +
u_2\) with \(u_1 \in W\) and \(u_2 \in W^{\perp}\), where the latter
is the orthogonal complement of \(W\). Clearly, \(\norm{u_1}^2 +
\norm{u_2}^2 = 1\). Further let \(\Delta \coloneqq \lambda_1 -
\lambda_2\) where \(\lambda_2\) is the second largest Eigenvalue of
\(X\) and for readibility let \(\lambda_1(Z \restriction W) \coloneqq
\max_{\norm{w}=1,  w \in W} w^\trans Z w\). We estimate
\begin{align*}
  u^\trans (X + \varepsilon Z)u &= u_1^\trans X u_1 + u_2^\trans X u_2
  + \varepsilon (u_1^\trans Z u_1 + u_1^\trans Z u_2 + u_2^\trans Z
  u_1 + u_2^\trans Z u_2) \\
& \leq \lambda_1 \norm{u_1}^2 + \lambda_2 \norm{u_2}^2 + \varepsilon
\lambda_1(Z \restriction W) + 3 \varepsilon \opnorm{Z} \norm{u_2} \\
& = \lambda_1 + \varepsilon \lambda_1(Z \restriction W) + (3 \varepsilon
\opnorm{Z} - \Delta \norm{u_2})\norm{u_2}\\
& = \lambda_1 + \epsilon
\lambda_1(Z \restriction W) - (\sqrt{\Delta}\norm{u_2}-3\eps\opnorm{Z}/\sqrt{4\Delta})^2 + 9\eps^2\opnorm{Z}^2/(4\Delta) \\
& \leq \lambda_1 + \epsilon \lambda_1(Z \restriction W) + 9\eps^2\opnorm{Z}^2/(4\Delta),
\end{align*}
which finishes the proof. 
\end{proof}

We record the following corollary of Lemma~\ref{lem:derlemma} for later use.
Recall that for a square matrix $X$, its \emph{exponential} is given by
\beqn
e^X = \sum_{k=0}^\infty \frac{1}{k!}X^k = I + X + \frac{1}{2}X^2 + \cdots.
\eeqn

\begin{corollary}\label{cor:derivative}
Let $r$ be a positive integer, $X\in \psd^r$ be a non-zero positive semidefinite matrices.  Let $\lambda_1 = \opnorm{X}$ and $W$ denote the $\lambda_1$-eigenspace of 
$X$. Then for $Z \in \R^{r\times r}$ symmetric, 
\[
\left.\frac{d_+}{d\eps}\opnorm{e^{\eps Z}Xe^{\eps Z}}\right|_{\eps = 0} = 2\lambda_1
\max_{\substack{w \in W \\ \norm{w}=1}} w^\trans  Z w
\] 
\end{corollary}
\begin{proof}
Let us write $e^{\eps Z} = \sum_{k=0}^\infty \frac{\eps^k Z^k}{k!} = I + \eps Z + \eps^2 R_\eps$, where $R_\eps = \sum_{k=2}^\infty
\frac{\eps^{k-2}Z^k}{k!}$. For $\eps < 1/(2\opnorm{Z})$, by the triangle inequality
\begin{align*}
\opnorm{R_\eps} \leq \sum_{k=2}^\infty \frac{\eps^{k-2}\opnorm{Z}^k}{k!} 
                \leq \frac{\opnorm{Z}^2}{2} \sum_{k=0}^\infty (\eps \opnorm{Z})^k = \frac{\opnorm{Z}^2}{2(1-\eps\opnorm{Z})} \leq \opnorm{Z}^2
\end{align*}
From here we see that
\[
e^{\eps Z}Xe^{\eps Z} = (I+\eps Z+ \eps^2 R_\eps)X(I + \eps Z + \eps^2 R_\eps) = X + \eps(ZX+XZ) + \eps^2(ZXR_\eps + R_\eps X Z + R_\eps X R_\eps)
\]
Let $R'_\eps = ZXR_\eps + R_\eps X Z + R_\eps X R_\eps$. Again by the triangle inequality, we have that
\[
\opnorm{R'_\eps} \leq 2\opnorm{Z}\opnorm{X}\opnorm{R_\eps} + \opnorm{R_\eps}^2\opnorm{X} \leq 2\opnorm{Z}^3\opnorm{X} + \opnorm{Z}^4\opnorm{X} = O(1) \text{,}
\]
for $\eps$ small enough. Therefore, we have that
\begin{align*}
\opnorm{e^{\eps Z}Xe^{\eps Z}} &= \opnorm{X + \eps(XZ+ZX) + \eps^2 R'_\eps} = \opnorm{X + \eps(XZ+ZX)} \pm O(\eps^2 \opnorm{R'_\eps}) \\
                             &= \opnorm{X + \eps(XZ+ZX)} \pm O(\eps^2) \text{.}
\end{align*}
Since $XZ+ZX$ is symmetric and $X \in \psd^r$ and non-zero, by Lemma~\ref{prop:standard} we have that
\begin{align*}
\opnorm{X+\eps(XZ+ZX)} &= \lambda_1 + \eps (\max_{\substack{w \in W \\ \norm{w}=1}} w^\trans (XZ+ZX)w) \pm O(\eps^2) \\
                     &= \lambda_1 + \eps \lambda_1(\max_{\substack{w \in W \\ \norm{w}=1}}
w^\trans (Z+Z)w) \pm O(\eps^2) \\
                     &= \lambda_1 + 2\lambda_1 \eps(\max_{\substack{w \in W \\ \norm{w}=1}} w^\trans Zw) \pm O(\eps^2) 
\end{align*}
Putting it all together, we get that
\[
\opnorm{e^{\eps Z}Xe^{\eps Z}} = \opnorm{X+\eps(XZ+ZX)} + O(\eps^2) = \lambda_1 +
2\lambda_1\eps(\max_{\substack{w \in W \\ \norm{w}=1}} w^\trans Zw) \pm O(\eps^2) 
\]
as needed.
\end{proof}

\section{Rescaling semidefinite factorizations}
\label{sec:resc-psd-extend}

A crucial point will be the rescaling of a semidefinite factorization of a
nonnegative matrix \(M\). In the case of linear extended formulations an upper bound of
\(\Delta\) on the largest entry of a slack matrix~\(S\) implies the existence of a minimal
nonnegative factorization \(S = UV\) where the entries of \(U,V\) are bounded by
\(\sqrt{\Delta}\). This ensures that the approximation of the extended formulation can be
captured by means of a polynomial-size (in~$\Delta$) grid. In the linear case, we note that any
factorization $S = UV$ can be rescaled by a nonnegative diagonal matrix $D$ where $S = (U
D)(D^{-1} U)$ and the factorization $(UD,D^{-1}V)$ has entries bounded by $\sqrt{\Delta}$.
However, such a rescaling relies crucially on the fact that after independent nonnegative
scalings of the coordinates a nonnegative vector remains nonnegative. However, in the
PSD setting, it is not true that the PSD property is preserved after independent
nonnegative scalings of the matrix entries. We circumvent this issue by showing that a
restricted class of transformations, i.e. the symmetries of the semidefinite cone, suffice
to rescale any PSD factorization such that the largest eigenvalue occurring in the
factorization is bounded in terms of the maximum entry in $M$ and the rank of the factorization.

\begin{theorem}[Rescaling semidefinite factorizations]
\label{thm:psdRescale}
Let $\Delta$ be a positive real number, $I,J$ be finite sets, $M \in[0,\Delta]^{I\times J}$ be a nonnegative matrix with a rank $r$ semidefinite
factorization factorization  $({\mathbf U}, {\mathbf V})$, ${\mathbf U} = (U_i)_{i\in I}$, ${\mathbf V} = (V^j)_{j\in J}$, satisfying $M_{ij} = \tr{U_i
V^j}$, $i \in I, j \in J$. Then there exists $A \in \psd^r$ such that $A{\mathbf U}A = (AU_iA)_{i \in I}$, $A^+{\mathbf V}A^+ = (A^+V^jA^+)_{j \in J}$ is a semidefinite factorization of $M$ satisfying
\begin{align*}
\lmax{A{\bf U}A} &= \max_{i\in I} \opnorm{AU_iA} \leq \sqrt{r\Delta} \\
\lmax{A^+{\bf V}A^+} &= \max_{j\in J}\opnorm{A^+V^jA^+}\leq \sqrt{r\Delta} \text{ .}
\end{align*}
\end{theorem}

\begin{proof}
Let $\bar{U} = \sum_{i \in I} U_i/|I|$, $\bar{V} = \sum_{j \in J} V^j/|J|$. Let $W_1 = {\rm Im}(\bar{U})$, $W_2 = {\rm Im}(\bar{V})$, $W =
P_{W_1}(W_2)$ and $d = \dim(W)$. Let $O \in \R^{r \times d}$ denote an orthonormal basis matrix for $W$, that is ${\rm Im}(O) = W$,
$OO^\trans = P_W$, and $O^\trans O = I_d$ (the $d \times d$ identity). 

As a first step, we preprocess the factorization to make it full
dimensional (i.e., by reducing the ambient dimension).

\begin{claim}\label{cl:resc-1} $(O^\trans {\mathbf U} O, O^\trans \mathbf{V} O)$ is a semidefinite factorization of $M$.
Furthermore, $O^\trans \bar{U} O$ and $O^\trans \bar{V} O$ are $d \times d$ nonsingular matrices.\end{claim}
\begin{claimproof}
If $T \in \psd^r$ then for any matrix $A \in \R^{r \times d}$, we have that $A^\trans T A \in \psd^d$. Hence $O^\trans U_i O,O^\trans V^jO \in \psd^d$, for all $i \in I,
j \in J$. To show that the new matrices factorize $M$, it suffices to show that $M_{ij} = \tr{O^\trans U_i O O^\trans V^j O}$ for all $i \in I,j \in J$. We examine
spectral decompositions of $U_i$ and $V_j$, 
\[
U_i = \sum_{k=1}^r \lambda_k u_k u_k^\trans \quad \text{ and } \quad V^j = \sum_{k=1}^r \gamma_k v_kv_k^\trans \text{.}
\]
For $k \in [r]$, we have that $u_k \in {\rm Im}(U_i) \subseteq \sum_{i \in I} {\rm Im}(U_i) = {\rm Im}(\bar{U}) = W_1$.
Similarly for $l \in [r]$, $v_l \in {\rm Im}(V^j) \subseteq {\rm Im}(\bar{V}) = W_2$. Given the previous containment,
remembering that $P_{W_1}(W_2) = W$, for $k,l \in [r]$ we have that
\[
\pr{u_k}{v_l} = \pr{P_{W_1}u_k}{v_l} = \pr{u_k}{P_{W_1} v_l} = \pr{u_k}{P_W v_l} = \pr{P_W u_k}{P_W v_l} = \pr{O^\trans u_k}{ O^\trans v_l} \text{,}
\]
since $O$ is an orthonormal basis matrix for $W$.  The trace inner product can now be analyzed as follows 
\begin{align*}
\tr{U_iV^j} &= \sum_{1 \leq k,l \leq r} \lambda_k \gamma_l \pr{u_k}{v_l}^2 
            = \sum_{1 \leq k,l \leq r} \lambda_k \gamma_l \pr{O^\trans u_k}{O^\trans v_l}^2 \\
            &= \tr{(O^\trans U_i O)(O^\trans V^j O)} \text{.}
\end{align*}
Hence $(O^\trans {\mathbf U} O, O^\trans {\mathbf V} O)$ is a semidefinite factorization of $M$ as needed. 

For the furthermore, we must show that the matrices $O^\trans \bar{U} O$ and $O^\trans \bar{V} O$ have trivial kernels. By construction ${\rm Ker}(\bar{U}) = W_1^\perp$, 
${\rm Ker}(\bar{V}) = W_2^\perp$, and 
\[
W = P_{W_1}(W_2) = (W_2 + W_1^\perp) \cap W_1 \text{.}
\]
From here, we have that
\[
\dim({\rm Ker}(O^\trans \bar{U} O)) = \dim({\rm Ker}(\bar{U} O)) = \dim(W_1^\perp \cap W) \leq \dim(W_1^\perp \cap W_1) = \dim(\{0\}) = 0 \text{.}
\]
Next, we have that
\begin{align*}
\dim({\rm Ker}(O^\trans \bar{V} O)) &= \dim({\rm Ker}(\bar{V} O)) = \dim(W_2^\perp \cap W) = \dim((W_2^\perp \cap W_1) \cap (W_2 + W_1^\perp)) \\
                              &= \dim((W_2 + W_1^\perp)^\perp \cap (W_2 + W_1^\perp)) = \dim(\{0\}) = 0 \text{,}
\end{align*}
as needed.
\end{claimproof} 

We will now examine factorizations of the form $(A O^\trans {\mathbf U} O A, A^{-1} O^\trans {\mathbf V} O A^{-1})$,
for $A \in \psd^d$ nonsingular. To see that this yields a factorization, note that
\begin{align*}
\tr{A O^\trans U_i O A A^{-1} O^\trans V^j O A^{-1}} = \tr{A O^\trans
  U_i O O^\trans V^j O A^{-1}}  \\ = \tr{O^\trans U_i O O^\trans V^j O A^{-1} A} 
                                                = \tr{O^\trans U_i O O^\trans V^j O } = M_{ij} \text{,}
\end{align*}
where the last inequality follows from Claim \ref{cl:resc-1}. To prove the theorem, it suffices to construct a nonsingular matrix $A \in
\psd^d$ such that 
\[
\lmax{A O^\trans {\bf U} O A} \leq \sqrt{d \Delta} \quad \text{ and }  \lmax{A^{-1} O^\trans {\bf V} O A^{-1}} \leq \sqrt{d \Delta} \text{.}
\]
Given such an $A$, we can recover the rescaling matrix claimed in the theorem using $O A O^\trans$, where $(O A O^\trans)^+ = O A^{-1} O^\trans$. It
is easy to check that this lifting is valid and preserves the maximum eigenvalues of the factorization matrices.

Given the above reduction, we may now assume that $d=r$ and that $\bar{U}$, $\bar{V}$ are nonsingular. We define the following potential function over
factorizations, 
\[
\Phi_M({\bf U},{\bf V}) = \lmax{{\bf U}}\cdot\lmax{{\bf V}} \text{.}
\]
We now examine the optimization problem 
\begin{equation}
\label{eq:resc-opt}
\inf_{\substack{A \in \psd^r \\ A \text{ nonsingular}}} \Phi_M(A{\bf U}A,A^{-1}{\bf V}A^{-1}) \text{ .}
\end{equation}
For any nonsingular $T \in \R^{r \times r}$, $(T{\bf U}T^\trans,
T^{-\trans}{\bf V}T^{-1})$ is a valid PSD factorization of
$M$. Without loss of generality we can require $T$ to be PSD as above,
since $T$ can be always be expressed as $T = O A$, where $O$ is
orthogonal and $A \in \psd^r$.  Here it is easy to
check that substituting $A$ for $T$ does not change the $\Phi_M$ value of the factorization.

Recall that the goal is to construct a nonsingular $A \in \psd^r$ such that 
\[
\lmax{A{\bf U}A} \leq \sqrt{r \Delta} \quad \text{ and } \quad \lmax{A^{-1}{\bf V}A^{-1}} \leq \sqrt{r \Delta} \text{ .}
\]
For any scalar $s > 0$, we see that 
\[
\lmax{sA{\bf U}sA} = s^2 \lmax{A{\bf U}A} \quad \text{ and } \quad \lmax{(sA)^{-1}{\bf V}(sA)^{-1}} = \lmax{A^{-1}{\bf V}A^{-1}}/s^2 \text{.}
\]
Given this, if $\Phi_M(A{\bf U}A,A^{-1}{\bf V}A^{-1}) \leq \mu^2$ then setting 
\[
s = \Phi_M(A{\bf U}A,A^{-1}{\bf V}A^{-1})^{1/4}/\lmax{A{\bf U}A}^{1/2} \text{,}
\]
we get that
\begin{align*}
\lmax{sA{\bf U}sA} = \lmax{(sA)^{-1}{\bf V}(sA)^{-1}} = \Phi_M(A{\bf U}A,A^{-1}{\bf V}A^{-1})^{1/2} \leq \mu \text{.}
\end{align*}
Hence it suffices to show that the infimum value for
\eqref{eq:resc-opt} is less than or equal to $r \Delta$.  We claim
that this infimum is attained. Since the objective functions is
clearly continuous in $A$, it suffices to show that the infimum can be
taken over a compact subset of $\psd^r$. Let $\tau = \Phi_M({\bf U},
{\bf V})$, and $\sigma > 0$ be the largest value such that $\bar{U}
\succeq \sigma I_r$, $\bar{V} \succeq \sigma I_r$. Note that $\sigma >
0$ exists since $\bar{U}$,$\bar{V}$ are nonsingular $r \times r$ PSD
matrices.

\begin{claim}\label{cl:resc-2} Let $A \in \psd^r$ nonsingular. If $\Phi_M(A{\bf U}A,A^{-1}{\bf V}A^{-1}) \leq \tau$, then there exists $s > 0$ such that 
$I_r \preceq sA \preceq (\tau/\sigma^2) I_r$.
\end{claim}
\begin{claimproof}
We examine the spectral decomposition of $A = \sum_{i=1}^r \lambda_i v_i v_i^\trans$, where $v_1,\dots,v_r$ form an orthornomal basis of $\R^r$ and
$\lambda_1 \geq \dots \geq \lambda_r \geq 0$. Note that $A^{-1} = \sum_{i=1}^r \lambda_i^{-1} v_i v_i^\trans$. Here $\lambda_r > 0$ since $A$ is
nonsingular.  Since multiplying $A$ by a positive scalar does not
change the potential $\Phi_M$, we may rescale $A$ such that $\lambda_r = 1$. Since $\lambda_r I_r
\preceq A \preceq \lambda_1 I_r$, and $\lambda_r = 1$, we must now show that $\lambda_1 \leq \tau/\sigma^2$.

We lower bound $\Phi(A)$ in terms of $\lambda_1$. Firstly, note that
\begin{align*}
\lmax{A{\bf U}A} &= \max_{i \in I} \opnorm{A U_i A} \geq \max_{i \in I} v_1^\trans A U_i A v_1 = \lambda_1 \max_{i \in I} v_1^\trans U_i v_1 \\
          &\geq \lambda_1 \frac{1}{|I|} \sum_{i \in I} v_1^\trans U_i v_1 = \lambda_1 v_1^\trans \bar{U} v_1 \geq \sigma \lambda_1 \text{.}
\end{align*}
Next, we have that
\begin{align*}
\lmax{A^{-1}{\bf V}A^{-1}} &= \max_{j \in J} \opnorm{A^{-1} V^j A^{-1}} \geq \max_{j \in J} v_r^\trans A^{-1} V^j A^{-1} v_r = \lambda_r^{-1} \max_{j \in J} v_r^\trans V^j v_r \\
          &= \max_{j \in J} v_r^\trans V^j v_r \geq \frac{1}{|J|} \sum_{j \in J} v_r^\trans V^j v_r = v_r^\trans \bar{V} v_r \geq \sigma \text{.}
\end{align*}
Therefore 
\[
\tau \geq \lmax{A{\bf U}A} \lmax{A^{-1}{\bf V}A^{-1}} \geq \lambda_1 \sigma^2 \Rightarrow \lambda_1 \leq \tau/\sigma^2 \text{,}
\]
as needed.
\end{claimproof}

From the above claim, and our assumption that $\Phi_M({\bf U}, {\bf V}) = \tau$, we get that
\[
\inf_{\substack{A \in \psd^r \\ A \text{ nonsingular}}} \Phi_M(A{\bf U}A, A^{-1}{\bf V}A^{-1})  = 
\inf_{\substack{A \in \psd^r \\ I_r \preceq A \preceq (\tau/\sigma^2) I_r}} \Phi_M(A{\bf U}A, A^{-1}{\bf V}A^{-1}) \text{.}
\] 
Since the infimum on the right hand side is taken on a compact set, the infimum is attained as claimed. Let $\mu^2$
denote the infimum value. Letting $(\widetilde{\mathbf U},\widetilde{\mathbf V}) = (A{\mathbf U}A, A^{-1}{\mathbf V}A^{-1})$, for the appropriate
matrix $A \in \psd^r$, we can assume
\[
\lmax{\widetilde{\mathbf U}} = \lmax{\widetilde{\mathbf V}} = \Phi_M(\widetilde{\mathbf U},\widetilde{\mathbf V})^{1/2} = \mu \text{.}
\]  
We shall now analyze how $\Phi_M$ behaves under small perturbations of the minimizer $(\widetilde{\mathbf U},\widetilde{\mathbf V})$.  Our goal is
to obtain a contradiction by assuming that~$\mu^2 > \Delta r + \tau$ for some~$\tau>0$.  To this end we bound the value of $\Phi_M$ at infinitesimal
perturbations of the point~$(\widetilde{\mathbf U},\widetilde{\mathbf V})$.  For a symmetric matrix~$Z$ and parameter~$\eps>0$ the type of
perturbations we consider are those defined by the invertible matrix~$e^{-\eps Z}$, which will take the role of the matrix~$A$ above.  Notice that
if~$Z$ is symmetric, then so is~$e^{-\eps Z}$.  We show that there exists a matrix~$Z$ such that for every $U\in\{\widetilde U_i\st i\in I\}$ such
that $\opnorm{U} = \mu$, we have
\beq\label{eq:Uvar}
\opnorm{e^{-\eps Z}U e^{-\eps Z}} \leq \mu  - \frac{2\mu}{r}\eps + O(\eps^2),
\eeq
while at the same time for every $V \in \{\widetilde V^j\st j\in J\}$ such that~$\opnorm{V} = \mu$, we have
\beq\label{eq:Vvar}
\opnorm{e^{\eps Z}V  e^{\eps Z}}  \leq \mu + \frac{2\Delta}{\mu}\eps + O(\eps^2).
\eeq
This implies that there is a point~$(\mathbf U',\mathbf V')$ in the neighborhood of the minimizer~$(\widetilde{\mathbf U},\widetilde{\mathbf V})$ where
\beqrn
\Phi_M(\mathbf U',\mathbf V') 
&\leq& 
\Big(\mu  - \frac{2\mu}{r}\eps + O(\eps^2) \Big)\cdot \Big(\mu + \frac{2\Delta}{\mu}\eps + O(\eps^2) \Big) \\
&=& \mu^2  - 2\Big( \frac{\mu^2}{r} - \Delta\Big)\eps + O(\eps^2)\\
&<& \mu^2 - \frac{2\tau}{r}\eps + O(\eps^2),
\eeqrn
where the last inequality follows from our assumption that~$\mu^2 > \Delta r + \tau$.
Thus, for small enough~$\eps>0$, we have~$\Phi_M(\mathbf U',\mathbf
V') <\mu^2$, a contradiction to the minimality of~$\mu$. It suffices
to consider the factorization matrices with the largest eigenvalues as
small perturbations cannot change the eigenvalue structure.
Hence, to prove the theorem we need to show the existence of such a matrix~$Z$.
\medskip

Let~$\mathcal Z\subseteq \sphere{r}$ be a finite set of unit vectors such that every~$z\in\mathcal Z$ is a $\mu$-eigenvector of at least one
of the matrices~$\widetilde U_i$ for $i\in I$. Let $p\in\R^\mathcal Z_+$ be a probability vector (i.e., $\sum_{z\in\mathcal Z}p(z) = 1$) and define the symmetric matrix
\beq\label{eq:Zdef}
Z = \sum_{z\in\mathcal Z}p(z)\, zz^\trans.
\eeq

\begin{claim}\label{clm:Vder}
Let $V \in \{\widetilde V^j\st j\in J\}$ be one of the factorization matrices such that $\opnorm{V} = \mu$.
Then,
\beq\label{eq:Vder}
\left.\frac{d_+}{d\eps}\opnorm{e^{\eps Z}V e^{\eps Z}}\right|_{\eps = 0} \leq \frac{2\Delta}{\mu}.
\eeq
\end{claim}

\begin{claimproof}
Let $\mathcal V\subseteq\sphere{r}$ be the set of eigenvectors of~$V$ that have eigenvalue~$\mu$.
Then, Corollary~\ref{cor:derivative} gives
\beq\label{eq:Vconv}
\left.\frac{d_+}{d\eps}\opnorm{e^{\eps Z}V e^{\eps Z}}\right|_{\eps = 0} 
= 2\mu\max_{v\in\mathcal V} v^{\trans}Zv = 2\mu \max_{v\in \mathcal V} \sum_{z \in \mathcal Z} p(z) (z^\trans v)^2
\eeq
We show that for any~$z\in\mathcal Z$ and~$v\in\mathcal V$, we have~$(z^\trans v)^2\leq \Delta/\mu^2$.
The claim  then follows from~\eqref{eq:Vconv} since~$p$ is a probability vector.
Let us fix vectors~$z\in\mathcal Z$ and~$v\in\mathcal V$ and let $U\in\{\widetilde U_i\st i\in I\}$ be a factorization matrix such that $z$ is a $\mu$-eigenvector of~$U$.  
Recall that the matrices~$U$ and~$V$ are part of a semidefinite factorization of the matrix~$M$ and that we assumed the entries of~$M$ to have value at most~$\Delta$.
Hence, $\Tr[U^\trans V] \leq \Delta$.
We now argue that~$\mu^2(z^\trans v) \leq \Tr[U^\trans V]$.
Let \(U = \sum_{k \in [r]}\lambda_ku_ku_k^{\trans}\) and $V = \sum_{\ell\in[r]}\gamma_\ell v_\ell v_\ell^{\trans}$ be spectral decompositions of~$U$ and~$V$,
respectively, such that~$u_1 = z$ and~$v_1 = v$.
The $\lambda_k$ and $\gamma_\ell$ are nonnegative (as~$U,V$ are PSD) and~$\lambda_1 = \gamma_1 = \mu$.
Hence, expanding the trace inner product
\beq\label{eq:exptr}
\Tr[U^\trans V] =  \sum_{k,\ell \in [r]} \lambda_k \gamma_\ell (u_k^\trans v_\ell)^2,
\eeq
we get that the terms on the right-hand side of~\eqref{eq:exptr} are nonnegative and that the sum in~\eqref{eq:exptr} is at least~$\lambda_1\gamma_1(u_1^\trans v_1)^2 = \mu^2(z^\trans v)^2$.
Putting these observations together we conclude that~$\mu^2(z^\trans v)^2\leq \Tr[U^\trans V] \leq\Delta$, which proves the claim.
\end{claimproof}

\begin{claim}\label{clm:Uder}
There exists a choice of unit vectors~$\mathcal Z$ and probabilities $p$ such that the following holds. Let $I' = \{i \in I \st \|\widetilde U_i\| = \mu\}$. Then, for~$Z$ as in~\eqref{eq:Zdef} we have
\beq\label{eq:Uder}
\left.\frac{d_+}{d\eps} \opnorm{e^{-\eps Z}\widetilde{U}_ie^{-\eps Z}}\right|_{\eps=0} \leq -\frac{2\mu}{r} \quad \forall i \in I'\text{.}
\eeq
\end{claim}

\begin{claimproof}
For every~$i\in I'$, let $\mathcal U_i\subseteq \R^r$ be the vector space spanned by the~$\mu$-eigenvectors of~$\widetilde U_i$.  Define the convex
set \(K = \conv{\bigcup_{i\in I'}(\mathcal U_i\cap B_2^r)}\). Notice that~$K$ is centrally symmetric. Let $k = \dim(K)$, and let $T \in \R^{r \times
k}$ denote a linear transformation such that that~$E = TB_2^k$ is the smallest volume ellipsoid containing~$K$. By John's Theorem, there exists a
finite set~$\mathcal Z\subseteq {\rm relbd}(K)\cap{\rm relbd}(E)$ and a probability vector~$p\in\R^\mathcal Z_+$ such that
\beq\label{eq:ZisT}
Z = \sum_{z\in\mathcal Z}p(z)\, zz^\trans = \frac{1}{k}{TT^{\trans}}.
\eeq
Notice that each~$z\in\mathcal Z$ must be an extreme point of~$K$ (as
it is one for~$E$) and the set of extreme points of~$K$ is exactly
\(\bigcup_{i\in I'} (\mathcal U_i\cap\sphere{r}).\)
Hence, each~$z\in\mathcal Z$ is a unit vector and at the same time a $\mu$-eigenvector of some~$\widetilde U_i$, $i\in I'$.

For $i \in I'$, by Corollary~\ref{cor:derivative} and~\eqref{eq:ZisT} we have that
\beqrn\label{eq:issue}
\left.\frac{d_+}{d\eps} \opnorm{e^{-\eps Z}\widetilde{U}_ie^{-\eps Z}}\right|_{\eps=0} 
&=& 
2\mu\max\{u^{\trans}(-Z)u\st u\in\mathcal U_i\cap \sphere{r}\}\\
&=& 
-2\mu\min\{u^{\trans}Zu\st u\in\mathcal U_i\cap \sphere{r}\}\\
&=&
-\frac{2\mu}{k}\min\{u^{\trans}TT^\trans u\st u\in\mathcal U_i\cap \sphere{r}\} \\
&\leq& -\frac{2\mu}{r}\min\{\norm{T^\trans u}^2 \st u\in\mathcal U_i\cap \sphere{r}\} \text{.}
\eeqrn
Since~$E\supseteq K\supseteq (\mathcal U_i\cap \sphere{r})$, for any~$u\in\mathcal U_i\cap\sphere{r}$, we have
\beqn
\norm{T^\trans u} = \sup_{x\in E} x^{\trans}u \geq \sup_{y\in K} y^{\trans}u \geq u^\trans u = 1 \text{ as needed.}
\eeqn
\end{claimproof}

Notice that the first claim implies~\eqref{eq:Vvar} and the second claim implies~\eqref{eq:Uvar}.
Hence, our assumption~$\mu^2 > \Delta r + \tau$ contradicts that~$\mu$ is the minimum value of~$\Phi_M$.
\end{proof}

\section{0/1 polytopes with high semidefinite xc}
\label{sec:exist-01-polyt}

The lower bound estimation will crucially rely on the fact that any 0/1
polytope in the \(n\)-dimensional unit cube can be written as a linear
system of inequalities \(Ax \leq b\) with integral coefficients where the largest coefficient is
bounded by \((\sqrt{n+1})^{n+1} \leq 2^{n\log(n)}\), see e.g.,
\cite[Corollary 26]{Ziegler:2000}. Using Theorem~\ref{thm:psdRescale} the proof follows along the lines of
\cite{Rothvoss11}; for simplicity and exposition we chose a
compatible notation. We use different estimation however and
we need to invoke Theorem~\ref{thm:psdRescale}. In the following let
\(\psd^r(\alpha) = \face{X \in \psd^r \mid \opnorm{X} \leq \alpha}\).

\begin{lemma}[Rounding lemma] 
\label{lem:roundingsdp}
For a positive integer $n$ set $\Delta \coloneqq
(n+1)^{(n+1)/2}$. Let $\mathcal X\subseteq\binSet^n$ be a nonempty
set, let $r \coloneqq \xcs(\conv{\mathcal X})$ and let $\delta \leq \big(16r^3(n+r^2)\big)^{-1}$.
Then, for every $i\in [n+r^2]$ there exist:
\begin{enumerate}
\item an integer vector $a_i\in\Z^n$ such that $\|a_i\|_\infty \leq \Delta$,
\item an integer~$b_i$ such that~$|b_i|\leq \Delta$,
\item a matrix~$U_i\in\psd^r(\sqrt{r\Delta})$ whose entries are
  integer multiples of~$\delta/\Delta$ and have absolute value at
  most~$8r^{3/2}\Delta$, such that
\end{enumerate}
\beqn
\mathcal X = \Big\{x\in\binSet^n\st\exists Y\in\psd^r(\sqrt{r\Delta}):\: \big|b_i -
a_i^{\trans}x - \langle Y, U_i\rangle\big| \leq \frac{1}{4(n+r^2)}\:\: \forall i\in [n+r^2]\Big\}.
\eeqn
\end{lemma}

\begin{proof}
For some index set~$I$ let $\mathcal A = (a_i,b_i)_{i\in I}\subseteq\Z^n\times \Z$ be a non-redundant description of~$\conv{\mathcal X}$ (i.e., $|I|$ is minimal) such that for every~$i\in I$, we have $\mnorm{a_i} \leq \Delta$ and $|b_i|\leq \Delta$.
Let $J$ be an index set for~$\mathcal X = (x_j)_{j\in J}$ and let
$S\in\Z_{\geq 0}^{I\times J}$ be the slack matrix of~$\conv{\mathcal
  X}$ associated with the pair~$(\mathcal A,\mathcal X)$. The largest
entry of the slack matrix is at most \(\Delta\). 
By Yannakakis's Theorem (Theorem~\ref{thm:yannaSDP}) there exists  a semidefinite factorization $(U_i, V^j)_{(i,j)\in I\times J}\subseteq\psd^r\times\psd^r$ of~$S$ such that
\beqn
\conv{\mathcal X} = \{x\in\R^n\st \exists Y\in\psd^r:\: a_i^{\trans}x + \langle U_i, Y\rangle = b_i\:\: \forall i\in I\}.  
\eeqn
By Theorem~\ref{thm:psdRescale} we may assume that ~$\opnorm{U_i}\leq\sqrt{r\Delta}$ for every $i\in I$ and $\opnorm{V^j}\leq\sqrt{r\Delta}$ for every
$j\in J$.  We will now pick a subsystem of maximum volume. For a linearly independent set of vectors $x_1,\dots,x_k \in \R^n$, we let $\vol{\{x_1,\dots,x_k\}}$ denote
the $k$-dimensional parallelepiped volume
\[
\vol{\sum_{i=1}^k a_i x_i \st a_1,\dots,a_k \in [0,1]} = \det((x_i^\trans x_j)_{ij})^{\frac{1}{2}} \text{.}
\]
If the vectors are dependent, then by convention the volume is zero. Let \(\mathcal W = \lspan{(a_i, U_i)\st i\in I}\) and let $I'\subseteq I$ be a subset of size 
$|I'| = \dim(\mathcal W)$ such that $\vol{\{(a_i, U_i)\st i\in I'\}}$ is maximized.  Note that~$|I'|\leq n+r^2$.

For any positive semidefinite matrix~$U\in \psd^r$ with spectral decomposition
\beqn
U = \sum_{k \in [r]}\lambda_k\, u_ku_k^{\trans}, 
\quad\text{we let}\quad
\bar U = \sum_{k \in [r]}\bar\lambda_k\, \bar u_k \bar u_k^{\trans}
\eeqn
be the matrix where for every $k\in [r]$, the value of~$\bar\lambda_k$ is the nearest integer multiple of~$\delta/\Delta$ to~$\lambda_k$ and $\bar
u_k$ is the vector we get by rounding each of the entries of~$u_k$ to the nearest integer multiple of~$\delta/\Delta$.  Since each $u_k$ is a unit
vector, the matrices~$u_ku_k^{\trans}$ have entries in~$[-1,1]$ and it follows that~$U$ has entries in~$r\opnorm{U}[-1,1]$.
Similarly, since each~$\bar u_k$ has entries in~$(1+\delta/\Delta)[-1, 1]$ each of the
matrices~$\bar u_k\bar u_k^{\trans}$ has entries in~$(1+\delta/\Delta)^2[-1, 1]$, and it
follows that~$\bar U$ has entries in~$r(\opnorm{U} + \delta/\Delta)(1+\delta/\Delta)^2[-1,1]$.
In particular, for every~$i\in I'$, the entries of~$\bar U_i$ are bounded in absolute value by
\beqn
r\big(\opnorm{U_i} + \delta/\Delta\big)(1 + \delta/\Delta)^2 \leq r(\sqrt{r\Delta} + \delta/\Delta)(1 + \delta/\Delta)^2 \leq 8r^{3/2}\sqrt{\Delta}.
\eeqn
We use the following simple claim.

\begin{claim}
Let $U$ and $\bar U$ be as above. Then, $\|\bar U - U\|_2 \leq 4\delta r^2/\sqrt{\Delta}$
\end{claim}
\begin{claimproof}
By the triangle inequality we have
\beqrn
\fnorm{\bar U - U}
&=& \fnorm{\sum_{k \in [r]}\bar\lambda_k\, \bar u_k \bar u_k^{\trans} -  \lambda_k\, u_ku_k^{\trans}}\\
&\leq& r\max_{k\in[r]}\fnorm{\bar\lambda_k\, \bar u_k \bar u_k^{\trans} -  \lambda_k\, u_ku_k^{\trans}}\\
&=& r\max_{k\in[r]}\fnorm{(\bar\lambda_k - \lambda_k)\, \bar u_k \bar u_k^{\trans} - \lambda_k(u_ku_k^{\trans} - \bar u_k \bar u_k^{\trans})} \\
&\leq& r\max_{k\in[r]} \frac{\delta}{\Delta}\fnorm{\bar u_k \bar u_k^{\trans}}  + \sqrt{r\Delta} \fnorm{u_ku_k^{\trans} - \bar u_k \bar u_k^{\trans}} \\
&=& r\max_{k\in[r]} \frac{\delta}{\Delta}\bar u_k^{\trans}\bar u_k + \sqrt{r\Delta}\fnorm{(u_k - \bar u_k)u_k^{\trans} - \bar u_k (\bar u_k^{\trans} - u_k^{\trans})} \\
&\leq& r\max_{k\in[r]} \frac{\delta}{\Delta}\Big(1 + \frac{\delta}{\Delta}\sqrt{r} \Big)^2 + \sqrt{r\Delta}\Big(\fnorm{u_k - \bar u_k} + \fnorm{\bar u_k}\fnorm{u_k - \bar u_k} \Big)\\
&\leq&
r \frac{\delta}{\Delta}\Big(1 +
  \frac{\delta}{\Delta}\sqrt{r} \Big)^2 + r
  \sqrt{r\Delta}\Big(\frac{\delta}{\Delta}\sqrt{r} + \big(1 +
  \frac{\delta}{\Delta}\sqrt{r}\big)\frac{\delta}{\Delta}\sqrt{r}\Big)\\
&\leq& r \cdot 4 \delta r / \sqrt{\Delta}.
\eeqrn
The claim now follows from the fact that $\delta\sqrt{r}/\Delta<1$.
\end{claimproof}

Define the set
\beqn
\bar{\mathcal X} = \Big\{x\in\binSet^n\st \exists Y\in\psd^r(\sqrt{r\Delta}):\: \big|b_i -
a_i^{\trans}x - \langle \bar{U_i}, Y\rangle\big|\leq \frac{1}{4(n+r^2)}\:\: \forall i\in I'  \Big\}.
\eeqn
We claim that~$\bar{\mathcal X} = \mathcal X$, which will complete the proof.

We will first show that $\mathcal X\subseteq\bar{\mathcal X}$. 
To this end, fix an index~$j\in J$.
By Theorem~\ref{thm:yannaSDP} we can pick~$Y = V^j\in \psd^r$ such that $a_i^{\trans}x_j + \langle U_i, Y\rangle = b_i$ for every $i\in I'$.
Moreover, $\opnorm{Y} = \opnorm{V^j} \leq \sqrt{r\Delta}$.
This implies that for every $i\in I'$, we have
\begin{multline*}
\big| b_i - a_i^{\trans}x_j - \langle \bar U_i, Y \rangle \big| 
=
\big| \underbrace{b_i - a_i^{\trans}x_j - \langle U_i, Y \rangle}_{0} + \langle \bar U_i - U_i, Y \rangle\big|\\ 
\leq \fnorm{\bar U_i - U_i}\fnorm{Y} \leq 4\delta r^3,
\end{multline*}
where the second line follows from the Cauchy-Schwarz inequality,
the above claim, and~$\fnorm{Y} \leq \sqrt{r}\opnorm{Y} \leq r \sqrt{\Delta}$.
Now, since $4\delta r^3  \leq 4r^3/(16r^3(n+r^2)) = 1/(4(n+r^2))$ we conclude that $x_j\in\bar{\mathcal X}$ and hence~$\mathcal X \subseteq \bar{\mathcal X}$.

It remains to show that $\bar{\mathcal X} \subseteq \mathcal X$. For this we show that
whenever $x\in\binSet^n$ is such that \(x \notin \mathcal X\) it follows that \(x \notin \bar{\mathcal X}\).
To this end, fix an~$x\in\binSet^n$ such that $x\not\in\mathcal X$.
Clearly~$x\notin\conv{\mathcal X}$ and hence, there must be an~$i^*\in I$ such that
$a_{i^*}^{\trans}x >b_{i^*}$.
Since~$x$, $a_{i^*}$ and $b_{i^*}$ are integral we must in fact have~$a_{i^*}^{\trans}x \geq b_{i^*}+1$.
We express this violation in terms of the above selected subsystem corresponding to the set~$I'$.

There exist unique multipliers $\nu\in\R^{I'}$ such that
\(\big(a_{i^*},U_{i^*}\big) = \sum_{i\in I'} \nu_i(a_i, U_i).\)
Observe that this  implies that
\(\sum_{i \in I'} \nu_i b_i = b_{i^*}\); otherwise it would be impossible for
\(a_i^{\trans}x + \langle U_i, Y\rangle = b_i\) to hold for every~$i\in I$ and hence we would have \(\mathcal{X} = \emptyset\) (which we assumed is not the case).

Using the fact that the chosen
subsystem \(I'\) is volume maximizing and using Cramer's rule, 
\[\card{\nu_i} = \frac{\vol{\face{(a_t,U_t) \mid t \in I' \setminus
    \face{i} \cup \face{i^*}}}}{\vol{\face{(a_t,U_t) \mid t \in I'
  }}} \leq 1.\]
For any \(Y \in \psd^r(\sqrt{r\Delta})\) using \(\langle U_{i^*}, Y\rangle \geq 0\) it follows thus
\begin{align*}
  1 &\leq \card{a_{i^*}^{\trans} x - b_{i^*} + \langle U_{i^*}, Y\rangle}  = \card{\sum_{i \in I'}
    \nu_i (a_{i}^{\trans} x - b_{i} + \langle U_{i}, Y\rangle)} \\
&  \leq \sum_{i \in I'}
    \card{\nu_i} \card{ a_{i}^{\trans} x - b_{i} + \langle U_{i}, Y\rangle} \leq (n+r^2)
    \max_{i\in I'}\card{ a_{i}^{\trans} x - b_{i} + \langle U_{i}, Y\rangle}.
\end{align*}

Using a similar estimation as above, for every $i\in I'$, we have
\begin{align*}
\card{ a_{i}^{\trans} x - b_{i} + \langle U_{i}, Y\rangle}
&=
|a_{i}^{\trans} x - b_{i} + \langle \bar U_{i}, Y\rangle + \langle U_i - \bar U_i, Y\rangle|\\
&\leq
|a_{i}^{\trans} x - b_{i} + \langle \bar U_{i}, Y\rangle | + |\langle U_i - \bar U_i, Y\rangle|\\
&\leq
|a_{i}^{\trans} x - b_{i} + \langle \bar U_{i}, Y\rangle | + \frac{1}{4(n+r^2)}.
\end{align*}
Combining this with \(1 \leq (n+r^2)
  \max_{i\in I'} \card{ a_{i}^{\trans} x - b_{i} + \langle U_{i}, Y\rangle}\) we obtain
\[\frac{1}{2(n+r^2)} \leq \frac{1}{n+r^2} - \frac{1}{4(n+r^2)} \leq \max_{i\in I'}\card{
a_{i}^{\trans} x - b_{i} + \langle \bar U_{i}, Y\rangle},\]
and so \(x \notin Y\).

Via padding with empty rows we can ensure that \(\card{I'} = n + r^2\)
as claimed. 
\end{proof}

Using Lemma~\ref{lem:roundingsdp} we can establish the existence of 0/1
polytopes that do not admit any small semidefinite extended formulation
following the proof of \cite[Theorem 4]{Rothvoss11}.

\begin{theorem} 
\label{thm:sdpNoEx}
For any \(n \in \N\) there exists \(\mathcal X \subseteq \binSet^n\) such that 
\[\xcs(\conv{\mathcal X}) = \Omega \left(\frac{2^{n/4}}{(n \log n)^{1/4}}\right ).\]
\end{theorem}
\begin{proof} 
Let \(R \coloneqq R(n) \coloneqq \max_{\mathcal X \subseteq \binSet^n}
\xcs(\conv{\mathcal X})\) and suppose that \(R(n) \leq 2^n\); otherwise the
statement is trivial. The construction of Lemma~\ref{lem:roundingsdp} induces
an injective map from \(\mathcal X \subseteq \binSet^n\) to systems \((a_i,
U_i,  b_i)_{i\in [n+r^2]}\) as the set \(\mathcal X\) can be reconstructed from
the system. Also, adding zero rows and columns to \(A, U\) and zero rows to
\(b\) does not affect this property. Thus without loss of generality we assume
that \(A\) is a \((n + R^2) \times n\) matrix, \(U\) is a \((n + R^2) \times
R^2\) matrix (using \(\frac{R(R+1)}{2} \leq R^2\)).  Furthermore, by
Lemma~\ref{lem:roundingsdp}, every value in \(U\) has absolute value at most
\(\Delta\) and can be chosen to be a multiple of \((16 R^3 (n +
R^2))^{-1}\Delta^{-1}\). Thus each entry can take at most $3(16 R^3 (n +
R^2))\Delta \cdot \Delta = \Delta^{2+o(1)}$ values, since $R \leq 2^n$ and
$\Delta \geq n^{n/2}$.  Furthermore, the entries of $A,b$ are integral and have
absolute value at most $\Delta$, and hence each entry can take at most $3\Delta
\leq \Delta^{2+o(1)}$ different values.
   
We shall now assume that $R \geq n$ (this will be justified by the lower bound on $R$ later).
By injectivity we cannot have more sets than distinct systems, i.e. 
\[2^{2^n} - 1 \leq \Delta^{(2+o(1))(n+R^2+1)(n+R^2)} = \Delta^{(2+o(1))R^4} = 2^{(2+o(1))n \log n R^4}\text{.}\]
Hence for $n$ large enough, \(R \geq  \frac{2^{n/4}}{(3n \log n)^{1/4}}\) as needed. 
\end{proof}

\section{On the semidefinite xc of polygons}
\label{sec:semid-xc-polyg}

In an analogous fashion to \cite{FioriniRothvossTiwary11} we can use a
slightly adapted version of Theorem~\ref{lem:roundingsdp} to show the
existence of a polygon with \(d\) integral vertices with semidefinite
extension complexity 
\(\Omega((\frac{d}{\log d})^{\frac{1}{4}})\). For
this we change Theorem~\ref{lem:roundingsdp} to work for arbitrary
polytopes with bounded vertex coordinates; the proof is almost
identical to Theorem~\ref{lem:roundingsdp} and follows with the analogous
changes as in \cite{FioriniRothvossTiwary11}. 

\begin{lemma}[Generalized rounding lemma] 
\label{lem:roundingsdpGen}
Let $n,N \geq 2$ be a positive integer and set $\Delta \coloneqq
((n+1)N)^{2n}$. Let $\mathcal V \subseteq\Z^{n} \cap [-N,N]^n$ be a nonempty and
convex independent set and \(\mathcal X \coloneqq \conv{\mathcal V}
\cap \Z^n\). With $r \coloneqq \xcs(\conv{\mathcal X})$ and  $\delta
\leq \big(16r^3(n+r^2)\big)^{-1}$,
for every $i\in [n+r^2]$ there exist:
\begin{enumerate}
\item an integer vector $a_i\in\Z^n$ such that $\|a_i\|_\infty \leq \Delta$,
\item an integer~$b_i$ such that~$|b_i|\leq \Delta$,
\item a matrix~$U_i\in\psd^r(\sqrt{r\Delta})$ whose entries are
  integer multiples of~$\delta/\Delta$ and have absolute value at
  most~$8r^{3/2}\Delta$, such that
\end{enumerate}
\beqn
\mathcal X = \Big\{x\in \Z^n\st\exists Y\in\psd^r(\sqrt{r\Delta}):\: \big|b_i -
a_i^{\trans}x - \langle Y, U_i\rangle\big| \leq \frac{1}{4(n+r^2)}\:\: \forall i\in [n+r^2]\Big\}.
\eeqn
\end{lemma}
\begin{proof}
By, e.g., \cite[Lemma D.4.1]{hindry2000diophantine} it follows that
\(P\) has a non-redundant description with integral coefficients of
largest absolute value of at most \(((n+1)N)^{n}\). Thus the maximal entry
occurring in the slack matrix is \(((n+1)N)^{2n} = \Delta\). The proof
follows now with a similar argument as in Theorem~\ref{lem:roundingsdp}.
\end{proof}

We are ready to prove the existence of a polygon with \(d\) vertices, 
with integral coefficients, so that its semidefinite extension
complexity is \(\Omega((\frac{d}{\log d})^{\frac{1}{4}})\). 

\begin{theorem}[Integral polygon with high semidefinite xc]
  For every \(d \geq 3\), there exists a \(d\)-gon \(P\) with
  vertices in \([2d] \times [4d^2]\) and \(\xcs(P) = \Omega((\frac{d}{\log d})^{\frac{1}{4}})\).
\end{theorem}
\begin{proof} 
The proof is identical to the one is \cite{FioriniRothvossTiwary11}
except for adjusting parameters as follows. The set \(Z \coloneqq
\face{(z,z^2) \mid z \in [2d]}\) is convex independent, thus every
subset \(X \subseteq Z\) of size \(\card{X} =d\) yields a different
convex \(d\)-gon. Let \(R \coloneqq \max\face{\xcs{\conv{X} \mid X
    \subseteq Z, \card{X} = d}}\). 

As in the proof of Theorem \ref{thm:sdpNoEx}, we need to count the number of
systems (which the above set of polygons map to in an injective manner). Using
\(\Delta = (12d^2)^2\), \(n = 2\), \(N = 4d^2\) by
Lemma~\ref{lem:roundingsdpGen} it follows easily that each entry in the system
can take at most \(cd^{14}\) different values. Without loss of generality, by
padding with zeros, we assume that the system given by
Lemma~\ref{lem:roundingsdpGen} has the following dimensions: the \(A,b\) part
from (1.)~and (2.), where \(A\) is formed by the rows \(a_i\), is a \((3 + R^2)
\times 3\) matrix and \(U\) from (3.), formed by the \(U_i\) read as rows
vectors, is a \((3 + R^2) \times R^2\) matrix. We estimate \[2^{d} \leq
(cd^{14})^{(3+R^2)^2} \leq 2^{c' \cdot R^4 \cdot \log d}\] and hence \(R \geq
c' (\frac{d}{\log d})^{\frac{1}{4}}\) for some constant \(c' > 0\) follows. 
\end{proof}

\section{Final remarks}
\label{sec:conclusion}

Most of the questions and complexity theoretic considerations in
\cite{Rothvoss11} as well as the approximation theorem carry over
immediately to our setting and the proofs
follow similarly. For example, in analogy to \cite[Theorem 6]{Rothvoss11},
an approximation theorem for 0/1 polytopes can be derived showing that
every semidefinite extended formulation for a 0/1 polytope can be
approximated arbitrarily well by one with coefficients of bounded
size.

 \medskip The following important problems remain open:

\begin{problem}
  Does the CUT polytope have high semidefinite extension
  complexity. We highly suspect that the answer is in the affirmative,
similar to the linear case. However the partial slack matrix analyzed in
  \cite{extform4} to establish the lower bound for linear EFs has an
  efficient semidefinite factorization. In fact, it was precisely this
  fact that   established the separation between semidefinite EFs and
  linear EFs in \cite{bfps2012}. 
\end{problem}

\begin{problem}
  Is there an information theoretic framework for lower bounding
  semidefinite rank similar to the framework laid out in \cite{braverman2012information,BP2013commInfo} for nonnegative rank? 
\end{problem}

\begin{problem}
  As asked in \cite{FioriniRothvossTiwary11}, we can ask similarly
  for semidefinite EFs: is the provided lower bound for the
  semidefinite extension complexity of polygons tight? 
\end{problem}

\section*{Acknowledgements}
\label{sec:acknowledgements}
We are indebted to the anonymous referees for their remarks and the
shortening of the proof of Lemma~\ref{lem:derlemma}.

\bibliographystyle{abbrvnat}
\bibliography{bibliography}

\begin{thebibliography}{18}
\providecommand{\natexlab}[1]{#1}
\providecommand{\url}[1]{\texttt{#1}}
\expandafter\ifx\csname urlstyle\endcsname\relax
  \providecommand{\doi}[1]{doi: #1}\else
  \providecommand{\doi}{doi: \begingroup \urlstyle{rm}\Url}\fi

\bibitem[Ben-Tal and Nemirovski(2001)]{Ben-TalNemirovski01}
A.~Ben-Tal and A.~Nemirovski.
\newblock On polyhedral approximations of the second-order cone.
\newblock \emph{Math. Oper. Res.}, 26:\penalty0 193--205, 2001.
\newblock \doi{10.1287/moor.26.2.193.10561}.

\bibitem[Braun and Pokutta(2013)]{BP2013commInfo}
G.~Braun and S.~Pokutta.
\newblock {Common information and unique disjointness}.
\newblock \emph{submitted}, 2013.

\bibitem[Braun et~al.(2012)Braun, Fiorini, Pokutta, and Steurer]{bfps2012}
G.~Braun, S.~Fiorini, S.~Pokutta, and D.~Steurer.
\newblock {Approximation Limits of Linear Programs (Beyond Hierarchies)}.
\newblock In \emph{53rd IEEE Symp. on Foundations of Computer Science (FOCS
  2012)}, pages 480--489, 2012.
\newblock ISBN 978-1-4673-4383-1.
\newblock \doi{10.1109/FOCS.2012.10}.

\bibitem[Braverman and Moitra(2012)]{braverman2012information}
M.~Braverman and A.~Moitra.
\newblock An information complexity approach to extended formulations.
\newblock \emph{Electronic Colloquium on Computational Complexity (ECCC)},
  19\penalty0 (131), 2012.

\bibitem[Faenza et~al.(2012)Faenza, Fiorini, Grappe, and
  Tiwary]{FaenzaFioriniGrappeTiwary11}
Y.~Faenza, S.~Fiorini, R.~Grappe, and H.~R. Tiwary.
\newblock Extended formulations, nonnegative factorizations, and randomized
  communication protocols.
\newblock In A.~Mahjoub, V.~Markakis, I.~Milis, and V.~Paschos, editors,
  \emph{Combinatorial Optimization}, volume 7422 of \emph{Lecture Notes in
  Computer Science}, pages 129--140. Springer Berlin Heidelberg, 2012.
\newblock ISBN 978-3-642-32146-7.
\newblock \doi{10.1007/978-3-642-32147-4_13}.
\newblock URL \url{http://dx.doi.org/10.1007/978-3-642-32147-4_13}.

\bibitem[Fiorini et~al.(2012{\natexlab{a}})Fiorini, Massar, Pokutta, Tiwary,
  and {de}~Wolf]{extform4}
S.~Fiorini, S.~Massar, S.~Pokutta, H.~R. Tiwary, and R.~{de}~Wolf.
\newblock {Linear vs. Semidefinite Extended Formulations: Exponential
  Separation and Strong Lower Bounds}.
\newblock \emph{{Proceedings of STOC 2012}}, 2012{\natexlab{a}}.

\bibitem[Fiorini et~al.(2012{\natexlab{b}})Fiorini, Rothvo{\ss}, and
  {Tiwary}]{FioriniRothvossTiwary11}
S.~Fiorini, T.~Rothvo{\ss}, and H.~R. {Tiwary}.
\newblock Extended formulations for polygons.
\newblock \emph{Discrete \& Computational Geometry}, 48\penalty0 (3):\penalty0
  658--668, 2012{\natexlab{b}}.
\newblock ISSN 0179-5376.
\newblock \doi{10.1007/s00454-012-9421-9}.
\newblock URL \url{http://dx.doi.org/10.1007/s00454-012-9421-9}.

\bibitem[Goemans(2009)]{Goemans09}
M.~X. Goemans.
\newblock {Smallest compact formulation for the permutahedron}.
\newblock
  \href{http://math.mit.edu/~goemans/PAPERS/permutahedron.pdf}{Manuscript},
  2009.

\bibitem[Gouveia et~al.(2011)Gouveia, Parrilo, and
  Thomas]{GouveiaParriloThomas2011}
J.~Gouveia, P.~A. Parrilo, and R.~Thomas.
\newblock Lifts of convex sets and cone factorizations.
\newblock \emph{Math. Oper. Res.}, 38\penalty0 (2):\penalty0 248--264, May
  2011.

\bibitem[Hindry and Silverman(2000)]{hindry2000diophantine}
M.~Hindry and J.~H. Silverman.
\newblock \emph{Diophantine geometry: an introduction}, volume 201.
\newblock Springer, 2000.

\bibitem[John(1948)]{john1948extremum}
F.~John.
\newblock Extremum problems with inequalities as subsidiary conditions.
\newblock In \emph{Studies and Essays presented to R. Courant on his 60th
  Birthday}, pages 187--204, 1948.

\bibitem[Kat{\=o}(1995)]{kato1995perturbation}
T.~Kat{\=o}.
\newblock \emph{Perturbation theory for linear operators}, volume 132.
\newblock springer, 1995.

\bibitem[Rothvo{\ss}(2012)]{Rothvoss11}
T.~Rothvo{\ss}.
\newblock Some 0/1 polytopes need exponential size extended formulations.
\newblock \emph{Math. Programming}, 2012.
\newblock \href{http://arxiv.org/abs/1105.0036}{arXiv:1105.0036}.

\bibitem[Rothvo{\ss}(2013)]{Rothvoss13}
T.~Rothvo{\ss}.
\newblock The matching polytope has exponential extension complexity.
\newblock \emph{ArXiv e-prints}, 2013.

\bibitem[Shannon(1949)]{Shannon49}
C.~E. Shannon.
\newblock The synthesis of two-terminal switching circuits.
\newblock \emph{Bell System Tech. J.}, 25:\penalty0 59--98, 1949.

\bibitem[Yannakakis(1988)]{Yannakakis88}
M.~Yannakakis.
\newblock Expressing combinatorial optimization problems by linear programs
  (extended abstract).
\newblock In \emph{Proc.\ STOC 1988}, pages 223--228, 1988.

\bibitem[Yannakakis(1991)]{Yannakakis91}
M.~Yannakakis.
\newblock Expressing combinatorial optimization problems by linear programs.
\newblock \emph{J. Comput. System Sci.}, 43\penalty0 (3):\penalty0 441--466,
  1991.
\newblock \doi{10.1016/0022-0000(91)90024-Y}.

\bibitem[Ziegler(2000)]{Ziegler:2000}
G.~M. Ziegler.
\newblock {Lectures on 0/1-Polytopes}.
\newblock In G.~Kalai and G.~M. Ziegler, editors, \emph{Polytopes ---
  Combinatorics and Computation}, volume~29 of \emph{DMV Seminar}, pages 1--41.
  Birkh\"{a}user Basel, 2000.
\newblock ISBN 978-3-7643-6351-2.
\newblock \doi{10.1007/978-3-0348-8438-9_1}.
\newblock URL \url{http://dx.doi.org/10.1007/978-3-0348-8438-9_1}.

\end{thebibliography}

\end{document}